\documentclass{amsproc}

\newtheorem{theorem}{Theorem}[section]
\newtheorem{lemma}[theorem]{Lemma}
\newtheorem{corollary}[theorem]{Corollary}
\newtheorem{conjecture}[theorem]{Conjecture}
\newtheorem{proposition}[theorem]{Proposition}

\theoremstyle{definition}
\newtheorem{definition}[theorem]{Definition}
\newtheorem{example}[theorem]{Example}

\theoremstyle{remark}
\newtheorem{remark}[theorem]{Remark}

\numberwithin{equation}{section}

\date{\today}

\newcommand{\Sp}{\mathbb{S}}
\newcommand{\R}{\mathbb{R}}

\begin{document}

\title{Log-optimal configurations on the sphere}

\author{P. D. Dragnev}
\address{Department of Mathematical Sciences, Indiana University-Purdue University Fort Wayne, Fort Wayne, Indiana 46805}
\email{dragnevp@ipfw.edu}
\thanks{This work is supported in part by Simons CGM no. 282207. The author would like to also thank Erwin Schr\"{o}dinger International Institute for its hospitality during his stay, when part of this manuscript was written.}



\dedicatory{This paper is dedicated to Ed Saff on the occasion of his 70-th Birthday.}


\begin{abstract}
In this article we consider the distribution of $N$ points
on the unit sphere $\mathcal{S}^{d-1}$ in ${\bf R}^d$ interacting via logarithmic potential. A characterization theorem of the stationary configurations is derived when $N=d+2$ and two new log-optimal configurations minimizing the logarithmic energy are obtained for six points on $\Sp^3$ and seven points on $\Sp^4$. A conjecture on the log-optimal configurations of $d+2$ points on $\Sp^{d-1}$ is stated and three auxiliary results supporting the conjecture are presented.
\end{abstract}

\maketitle

%

\section{Introduction and main results}
\setcounter{equation}{0}

Minimal energy configurations have wide ranging applications in various fields of science, such as crystallography, nanotechnology, material science, information technology, wireless communications, complexity of algorithms, etc. In the last twenty years Ed Saff has been one of the leaders in the field. His contributions are numerous and his enthusiasm for the subject contagious, as experienced first-hand by the author himself.

In this contribution we shall characterize the stationary configurations (or "ground states") of the discrete logarithmic energy on $\Sp^{d-1}$ when $N=d+2$. As a consequence of our characterization theorem, we will present a simplified proof of the log-optimality on $\Sp^2$ of the bipyramid, the first so-called non-sharp configuration in dimension $d=3$ as defined by Cohn and Kumar in \cite{CK}. We shall also derive rigorously two new log-optimal configurations, six points on $\Sp^3$ and seven points on $\Sp^4$, which are the first non-sharp configurations in dimensions $d=4$ and $d=5$ respectively. This leads us to state a conjecture on the log-optimal configuration of $d+2$ points on $\Sp^{d-1}$, supported by three auxiliary results used in establishing the log-optimality of the aforementioned configurations. Some of the results were previously announced in the extended abstract \cite{D}. Here we provide all of the proofs.

For every ${\bf x}, {\bf y} \in \mathbb{R}^d $ let ${\bf x}\cdot {\bf y}=x_1 y_1+\dots+x_d y_d$ be the inner product and $|{\bf x} |=({\bf x}\cdot {\bf x})^{1/2}$ the Euclidean distance. Denote the unit sphere with $\mathbb{S}^{d-1}:=\{ {\bf x}\in \mathbb{R}^d : |{\bf x}|=1\}$. For any
$N$-point configuration $\omega_N = \{ {\bf x}_1,{\bf x}_2, \dots, {\bf x}_N \} \subset \mathbb{S}^{d-1}$ the points $\{ {\bf x}_i \}$ and the segments $\{ {\bf x}_i{\bf x}_j \}_{i\not= j}$ will be called respectively {\it vertices} and {\it edges} of the configuration. Throughout $d_{i,j}:=|{\bf x}_i - {\bf x}_j|^2$ will denote the square of the length of the corresponding edge. Here we are interested in configurations of points $\omega_N^*$ such that

\begin{equation}\label{prod}
P (\omega_N^* )={\mathcal P}(N,d):=\max_{\omega_N \subset \mathbb{S}^{d-1} } P(\omega_N ), \quad
P (\omega_N ) =\prod_{1\leq i<j\leq N} d_{i,j}
\end{equation}

\noindent These configurations minimize the logarithmic energy

\begin{equation}\label{LogEn}
E_{\rm log} (\omega_N ):= \sum_{1\leq i \not= j\leq N} \log
\frac{1}{|{\bf x}_i-{\bf x}_j|}=- \log P(\omega_N),
\end{equation}

\noindent and hence are called {\it log-optimal} configurations.

More generally, the {\it optimal Riesz $s$-energy} configurations minimize (maximize when $s<0$) the {\it $s$-energy}
\begin{equation}\label{s-energy}
\mathcal{E}_s (N,d) := \min_{\omega_N \subset {\Sp}^{d-1}} E_s (\omega_N) ,
\end{equation}
where
\[
E_s (\omega_N ):= \sum_{1\leq i \not= j\leq N}
\frac{1}{|{\bf x}_i-{\bf x}_j|^s} .
\]
The log-optimal configurations are the limiting case of the optimal $s$-energy configurations as $s \to 0$. As $s \to \infty $ we arrive at {\it best packing configurations} (centers of $N$ identical spherical caps with maximal radius packed on the unit sphere, a problem referred to as {\em Tammes problem}). For further reference on discrete minimal energy problems see \cite{CK}, \cite{ConSlo}, \cite{HS}, \cite{MKS}, \cite{RSZ1}, \cite{SK1}.

The question of finding log-optimal configurations was posed by Whyte in 1952 \cite{W} (for $d=3$), yet only very few are known. If $d=2$ this is a well studied problem of Fekete points on the circle, and the solution is the regular $N$-gon. If $d=3$, log-optimal points are referred to as elliptic Fekete points and the solution is known for $N=1-6$, and $12$. For $N=1-4$ the solution is trivial (regular simplex of dimension $0-3$). For $N=12$ Andreev \cite{A} showed that the regular icosahedron is an optimal configuration. He used the fact that the configuration is a spherical 5-design, a technique that follows closely the results of Kolushov and Yudin \cite{KY}, where the authors provided a lower bound for $ {\mathcal P} (N,d)$, which they showed is attained for the regular $d$-simplex when $N=d+1$, and the generalized octahedron when $N=2d$. The analysis is based on the fact that these special configurations are suitable spherical designs (for definition of spherical designs see \cite{DGS}). Cohn and Kumar \cite{CK} subsequently defined such designs, having $m$ distinct inner products and strength $2m-1$, as {\em sharp configurations}, and proved that all sharp configurations are universally optimal, i.e. optimal for a large class of potential energies that includes the Riesz $s$-energy case.

However, as noted by Kolushov and Yudin for the particular case when $N=5$ and $d=3$, when the optimal configuration is not a design of sufficiently high degree, then the method fails. Indeed, for  $\Sp^2$ the first case when sharp configuration doesn't exist is when $N=5$.
The optimal configurations in this case are rigorously found in only few cases. The log-optimal solution was found in \cite{DLT}, and is the triangular bipyramid, i.e. two points at the Poles and three points forming an equilateral triangle on the Equator. In \cite{HoS} Hou and Shao utilize computer-aided proof to show that the bipyramid maximizes the sum of all distances (so-called {\em Fejes-T\'{o}th problem}) between five points on the sphere $\Sp^2$. Richard Schwartz solved the {\em Thomson's problem} $s=1$, as well as the case $s=2$ by also employing complex computational methods. The solution is also the bipyramid. However, in a very elegant paper Bondarenko, Hardin and Saff \cite{BHS} derived that as $s\to \infty$ any limiting configuration of Riesz $s$-energy optimal $5$-point configurations must be a square-based best-packing pyramid (North Pole and a square on the Equator). Indeed, it is an easy calculus problem to see that a certain square-based pyramid will have smaller energy for large enough $s$. It is a long standing conjecture by Melnik, Knop and Smith \cite{MKS} that for $s\leq s_0\approx 15.048...$ the bipyramid minimizes the Riesz $s$-energy, while for $s>s_0$ a square-based pyramid with altitude of the square base depending on $s$ is the minimizer. We should point out that in a recent paper Tumanov \cite{T} shows without computer use that the bipyramid has minimal biquadratic energy.

Our goal here is to characterize the {\em stationary configurations}, also called {\em ground states}, for which the gradient of the discrete logarithmic energy \eqref{LogEn} vanishes. The following vector equations, describing the ground states are found in \cite{BBP} for $d=3$ and \cite{DLT} for $d>3$.

\vskip 1mm

\begin{proposition} Let $\omega_N=\{ {\bf x}_1, {\bf x}_2, \dots , {\bf x}_N\}$ be a stationary logarithmic configuration on the unit sphere $\Sp^{d-1}$ in ${\R}^d$. Then the following force conditions hold:
\begin{equation}\label{VecEqs}
\sum_{j\not= i} \frac{{\bf x}_i- {\bf x}_j}{d_{i,j}} = \frac{N-1}{2}\,
{\bf x}_i\ \ \ i=1,\dots, N.
\end{equation}
Moreover, the center of mass of the configuration coincides with the center of the sphere ${\bf 0}$ and
\begin{equation}\label{sum_2N}
\sum_{j\not= i} d_{i,j} = 2N \ \ \ i=1,\dots,N.
\end{equation}
\end{proposition}
\vskip 1mm
In general, stationary $s$-energy configurations satisfy similar vector equations, but the coefficients on the right-hand side of \eqref{VecEqs} vary with $i$, which adds significant difficulty.
\vskip 1mm

To formulate our characterization theorem we introduce some notions from dimension theory. Given $k$ points ${\bf x}_1,\dots, {\bf x}_k$ in ${\R}^d$, let ${\bf u}_i:= {\bf x}_i-{\bf x}_1$, $i=2,\dots,k$. The hyperplane spanned by the points $\{ {\bf x}_i\}$ is the set
\begin{equation}\label{hyperplane}
G_{\{ {\bf x}_i \}} :=\{ {\bf w} \in {\R}^d\ : \ {\bf w}={\bf x}_1+\alpha_2 {\bf u}_2+\cdots +\alpha_k {\bf u}_k,\ \alpha_i \in {\R}\}.\end{equation}
Clearly, the choice of ${\bf x}_1$ here is not restrictive and can be substituted with any other point in the hyperplane. The dimension of this hyperplane is the dimension of ${\rm span}\{ {\bf u}_2,\dots,{\bf u}_k\}$. Observe that the span of the points $\{ {\bf x}_i\}$ is not the same as the span of the vectors $\{ {\bf x}_i\}$ (unless ${\bf 0}\in G_{\{ {\bf x}_i \}}$).

\begin{definition}\label{mirror}
Let $\omega_N =\{ {\bf x}_1,\dots,{\bf x}_N\} \subset \mathbb{S}^{d-1}$. Two vertices ${\bf x}_i$ and ${\bf x}_j$ are called {\it mirror related} (we write ${\bf x}_i \sim {\bf x}_j$), if $d_{i, k}=d_{j, k}$, for every $k \not= i,j$. A configuration is called {\em degenerate} if the points of the configuration do {\sl not} span the whole ${\R}^d$.
\end{definition}

\begin{remark}\label{mirror_remark}
Observe that if ${\bf x}_i \sim {\bf x}_j$, then the hyperplane spanned by $\omega_N\setminus\{ {\bf x}_i, {\bf x}_j \}$ is contained in the orthogonal bisector hyperspace of the segment ${\bf x}_{i}  {\bf x}_{j}$. The points ${\bf x}_{i}$ and ${\bf x}_{j}$ are then mirror images of each other with respect to this hyperspace. This explains our choice of terms.
\end{remark}

\begin{proposition}\label{equivalence}
The mirror relation property in Definition \ref{mirror} is an equivalence relation.
\end{proposition}
\begin{proof} We only need to show the transitivity property, namely that ${\bf x}_{i} \sim {\bf x}_{j}$ and ${\bf x}_{j} \sim {\bf x}_{k}$ implies ${\bf x}_{i} \sim {\bf x}_{k}$. Indeed, if $s\not= i,j,k$, then $d_{s,i}=d_{s,j}$ and $d_{s,j}=d_{s,k}$ shows that $d_{s,i}=d_{s,k}$. For $s=j$ we have $d_{j,i}=d_{i,k}$ from ${\bf x}_{j} \sim {\bf x}_{k}$, and $d_{j,k}=d_{i,k}$ from ${\bf x}_{i} \sim {\bf x}_{j}$, which yields the transitivity. Moreover, if ${\bf x}_{i} \sim {\bf x}_{j}\sim {\bf x}_{k}$, then $d_{i,j}=d_{i,k}=d_{j,k}$, which can be generalized to make the important conclusion that a collection of points in an equivalence class forms a regular simplex.
\end{proof}

Next we formulate our characterization theorem.
\begin{theorem}\label{characterization}
Let $N=d+2$ and let the configuration $\omega_N$ be stationary. Then at least one of the following three possibilities occurs:
\begin{itemize}
\item[(a)] The configuration $\omega_N $ is degenerate;
\item[(b)] There exists a vertex with all edges stemming out being equal;
\item[(c)] Every vertex is mirror related to another vertex.
\end{itemize}
\end{theorem}

The following strict monotonicity property of $\mathcal{P} (N,d)$ shows that degenerate stationary logarithmic configurations are not log-optimal when $N\geq d+1$.
\vskip 1mm

\begin{theorem}\label{monotonicity}
{\it For fixed $N$, the sequence $ \mathcal{P} (N,d)$ is strictly increasing for $d<N$ and
$\mathcal{P} (N,d)=\mathcal{P} (N,N-1)$ for $d\geq N$.}
\end{theorem}
As is seen from the proof, it can be adapted to cover more general potential interaction, for example, the same is true for Riesz $s$-energy optimal points.

Next, we illustrate Theorem \ref{characterization} with the following classifications of the stationary configurations of $d+2$ points for dimensions $d=2$ and $d=3$.

\begin{example} Let $d=2$ ($N=4$). Then (a) and (b) are impossible and the only stationary configurations satisfy (c). There could be only two equivalence classes of two points each, which are easily seen to be the diagonals of a square.
\end{example}
\begin{corollary}[\cite{DLT}] \label{bipyramidOpt} The bipyramid is the unique up to rotation log-optimal configuration on $\Sp^2$.
\end{corollary}
\begin{proof}
In this case all possibilities (a), (b) and (c) are occurring. The only degenerate stationary configuration $\omega_{5}^{r}$ is the regular pentagon. The only stationary configuration  $\omega_5^p$ satisfying (b) is the square pyramid with vertex at the North Pole and a square base in a horizontal plane of altitude $-1/4$.

If (c) holds, there could be only two equivalence classes, one with two points, a segment, and the other with three points, an equilateral triangle, which we orient horizontally. The two points from the segment have to be equidistant to the vertices of the equilateral triangle, so clearly they are the North and South Poles. The center of mass shows that the equilateral triangle lies on the equator. Comparing the energies we observe that the bipyramid configuration $\omega_5^b$ minimizes the energy \eqref{LogEn}, which is another proof of the result in \cite{DLT}.
\end{proof}

We note that numerical evidence supports the conjecture of Melnik et al. that the triangular bipyramid configuration $\omega_5^b$ is minimizing the $s$-energy for $s<15.048...$, while the square pyramid $\omega_5^{p,s}$ is optimal (the base altitude is adjusted with $s$). Remarkably, they are the two competing stationary configurations above.

We next present two log-optimal configurations that are new in the literature.

\begin{theorem}\label{log-optimal} (i) The log-optimal configuration on $\Sp^3$ is unique up to rotation and is given by two orthogonal equilateral triangles (simplexes) inscribed in great circles.
\begin{equation}\label{six}
\omega_{\{3,3\}}:=\left\{(\cos\frac{2k\pi}{3},\sin\frac{2k\pi}{3},0,0)\right\}_{k=0}^2 \cup \left\{(0,0,\cos\frac{2k\pi}{3},\sin\frac{2k\pi}{3})\right\}_{k=0}^2.\end{equation}

(ii) The log-optimal configuration on $\Sp^4$ is unique up to rotation and is given by two orthogonal simplexes, an equilateral triangle and a regular tetrahedron, inscribed in a great circle and a great $3$-D hypersphere.
\begin{equation}\label{seven}
\begin{split}
\omega_{\{3,4\}}:=
&\displaystyle{\left\{(\cos\frac{2k\pi}{3},\sin\frac{2k\pi}{3},0,0,0)\right\}_{k=0}^2 \cup} \\
& \displaystyle{\left\{(0,0,0,0,1),(0,0,\frac{2\sqrt{2}}{3}\cos\frac{2k\pi}{3},\frac{2\sqrt{2}}{3}\sin\frac{2k\pi}{3},-1/3)\right\}_{k=0}^2.}
\end{split}
\end{equation}
\end{theorem}

Based on this theorem and the auxiliary lemmas in section 3 we state the following conjecture.
Let $[\, \cdot \, ]$ be the greatest integer function.

\begin{conjecture} \label{conjecture} The log-optimal configuration of $d+2$ points on $\Sp^{d-1}$  is unique up to rotations and consists of two mutually orthogonal regular simplexes, a $[ d/2 ]$-simplex  and a
$[ (d+1)/2 ]$-simplex respectively, denoted as $\omega_{\{[ d/2 ],[ (d+1)/2 ]\}}$ The maximal product from \eqref{prod} is given by

\begin{equation}\label{conj1}
\mathcal{P}(d+2,d)= \displaystyle{2^{\frac{(d+2)(d+1)}{2}}}
\displaystyle{\bigg(\frac{[\frac{d}{2} ]+1}{[\frac{d}{2} ]}\bigg)^
{([\frac{d}{2} ]+1)([\frac{d}{2}])/2}}
\displaystyle{\bigg(\frac{[\frac{d+1}{2} ]+1}{[\frac{d+1}{2} ]}\bigg)^
{([\frac{d+1}{2} ]+1)([\frac{d+1}{2}])/2}}.
\end{equation}
\end{conjecture}

In the next section we include the proof of the characterization Theorem \ref{characterization}, as well as the monotonicity Theorem \ref{monotonicity}. In section 3 we formulate and prove three lemmas utilized in the proof of the log-optimality of the configurations \eqref{six} and \eqref{seven}. Since they are important in their own right and support our Conjecture \ref{conjecture}, we choose to formulate them in the greatest generality and to include them in a separate section. The proof of the log-optimality of the two new configurations is included in the last section.

\section{Proof of the characterization and monotonicity theorems}
\setcounter{equation}{0}
We first start with the proof of the Characterization Theorem.
\begin{proof}[Proof of Theorem \ref{characterization}]
Suppose that $\omega_N$ is not degenerate. We have to show that (b) or (c) must be true. If (c) doesn't hold, then there is a vertex, say ${\bf x}_{N}$, which is {\sl not} mirror related to any other vertex. We will show that (b) holds in this case.

Without loss of generality we can assume that ${\bf x}_{N} = (0,0,\dots,1)$. Let ${\bf x}_{i} =({\bf y}_{i},r_i)$, where ${\bf y}_{i}\in {\bf R}^{d-1}$ and $r_i \in {\bf R}$, $i=1,\dots,N$. Then the stereographical projection with pole ${\bf x}_{N}$ of ${\bf x}_{i}$ on the hyperplane $\{ x_d =0 \}$ is given by ${\bf a}_{i}={\bf y}_{i}/(1-r_i)$, $i=1,\dots,N-1$. After we rewrite \eqref{VecEqs} in terms of $\{ {\bf y}_{i}\}$ and $\{ r_i \}$, we get

\begin{eqnarray}
-\sum_{j\not= i,N} \frac{{\bf y}_{j}}{d_{i,j}} &=& \bigg( \frac{N-1}{2}-\sum_{j\not= i} \frac{1}{d_{i,j}}\bigg)\,
{\bf y}_{i}\ \ \ i=1,\dots,N\label{VecEq1}\\
\sum_{j\not= i,N} \frac{1-r_j}{d_{i,j}} &=& \frac{N-1}{2}\, r_i + (1-r_i) \sum_{j\not= i} \frac{1}{d_{i,j}}\ \ \ i=1,\dots,N \label{VecEq2}
\end{eqnarray}
Observe that $d_{N, i}=2(1-r_i),\ i=1,\dots ,N-1$. From \eqref{VecEq2} we get that
\begin{equation}\label{CoeffEq}
d_{N, i}\bigg( \frac{N-1}{2}-\sum_{j\not= i} \frac{1}{d_{i,j}}\bigg) = N-1-\sum_{j\not= i,N} \frac{d_{N, j}}{d_{i,j}}  \ \ \ i=1,\dots,N,
\end{equation}
which coupled with \eqref{VecEq1} gives
\begin{equation}\label{a_VecEq}
\bigg( N-1-\sum_{j\not= i} \frac{d_{N, j}}{d_{i,j}}\bigg) {\bf a}_{i}
+\sum_{j\not= i,N} \frac{d_{N, j}}{d_{i,j}}\, {\bf a}_{j} =0
\ \ \ i=1,\dots,N-1.
\end{equation}
When $i=N $ we simply obtain
\begin{equation}\label{a_Mass}
\sum_{j=1}^{N-1} {\bf a}_{j} =0,
\end{equation}
which means that ${\bf 0}$ remains a center of mass for $\{ {\bf a}_{j} \}$.
The center of mass condition ${\bf x}_{1}+\cdots+{\bf x}_{N}=0$ translates to
\begin{equation}\label{a_Mass2}
\sum_{j=1}^{N-1} d_{N,j} {\bf a}_{j} =0.
\end{equation}

The vector equations \eqref{a_VecEq}, \eqref{a_Mass}, and \eqref{a_Mass2} can be written in matrix form as $MA=0$, where
$A$ is the $(N-1)\times(d-1)$ matrix with row-vectors $\{{\bf a}_{i}\}$, and $M$ is the $(N+1)\times (N-1)$ coefficient matrix
\begin{equation}\label{matrix}
M=
\left(
\begin{array}{cccc}
\scriptstyle{1} & \scriptstyle{1} & \dots & \scriptstyle{1} \\
& & & \\
\scriptstyle{d_{N,1}} & \scriptstyle{d_{N,2}} &\cdots & \scriptstyle{d_{N,N-1}} \\
& & & \\
\scriptstyle{(N-1)}-\sum_{j\not=1,N} \frac{d_{N,j}}{d_{1,j}}
& \frac{d_{N,2}}{d_{1,2}} & \cdots & \frac{d_{N,N-1}}{d_{1,N-1}} \\
& & & \\
\vdots
& \vdots
& \dots & \vdots\\
\frac{d_{N,1}}{d_{N-1,1}} & \frac{d_{N,2}}{d_{N-1,2}} & \cdots &
\scriptstyle{(N-1)}-\sum_{j\not=N-1,N} \frac{d_{N,j}}{d_{N-1,j}} \\
\end{array}
\right)
\end{equation}

Since $\omega_N$ is non-degenerate, the vectors $\displaystyle{\{{\bf a}_{i}\}_{i=1}^{N-1}}$ span all of $\{ x_d = 0\}$, so ${\rm rank}(A)=d-1$. This implies that ${\rm ker}(M)\geq d-1$, or
\begin{equation} \label{rank_condition} {\rm rank}(M)\leq N-d.\end{equation}
We point out that this property of the matrix $M=M(\{d_{i,j}\})$ holds in general for {\em any critical non-degenerate configurations}.

In our case $N-d=2$, hence for any $i$ we have
\begin{equation}
{\rm rank}
\left(
\begin{array}{ccccc}
\scriptstyle{1} & \dots & \scriptstyle{1} & \dots & \scriptstyle{1} \\
& & & \\
\scriptstyle{d_{N,1}} & \cdots & \scriptstyle{d_{N,i}} &\cdots & \scriptstyle{d_{N,N-1}} \\
& & & \\
\frac{d_{N,1}}{d_{i,1}} & \cdots & \scriptstyle{(N-1)}-\sum_{j\not=i,N} \frac{d_{N,j}}{d_{i,j}}
& \cdots & \frac{d_{N,N-1}}{d_{i,N-1}} \\
\end{array}
\right)\leq 2,\ i=1,...,N-1.
\end{equation}
If the rank above is $1$, then $d_{N,1}=\dots=d_{N,N-1}$ and (b) holds. So, we may assume that the rank is $2$. We now fix $i$ and substitute the $i$-th column with the sum of all columns, then multiply the $j$-th column of the resulting matrix with $d_{i,j}$ for all $j\not=i,N$. The new matrix will have the same rank.
\[
{\rm rank}
\left(
\begin{array}{ccccc}
\scriptstyle{d_{i,1}} & \cdots & \scriptstyle{N-1} &\cdots & \scriptstyle{d_{i,N-1}} \\
& & & & \\
\scriptstyle{d_{N,1} d_{i,1}} & \cdots & \scriptstyle{2N}& \cdots & \scriptstyle{d_{N,N-1} d_{i,N-1}} \\
& & & & \\
\scriptstyle{d_{N,1}} & \cdots & \scriptstyle{N-1} &\cdots & \scriptstyle{d_{N,N-1}} \\
\end{array}
\right)=2.
\]
If we fix some $j\not= i$, substitute the $k$-th column ($k\not= i,j$) with the sum of all the columns but the $i$-th, the rank of the resulting matrix will still be $2$. This implies that the $3\times 3$ determinant made of $i,j,k$-th columns will be zero. Using \eqref{sum_2N} we get that
\begin{equation}
\det
\left(
\begin{array}{ccc}
\scriptstyle{d_{i,j}} &  \scriptstyle{N-1} & \scriptstyle{2N-d_{N,i}} \\
& &  \\
\scriptstyle{d_{N,j} d_{i,j}} & \scriptstyle{2N}& \sum_{l \not= i,N} \scriptstyle{d_{N,l} d_{i,l}} \\
& &  \\
\scriptstyle{d_{N,j}} & \scriptstyle{N-1} & \scriptstyle{2N-d_{N,i}} \\
\end{array}
\right)= 0,
\end{equation}
which reduces to
\begin{equation}\label{DetEq}
(d_{i,j}-d_{N,j})\bigg[ 2N(2N-d_{N,i})-(N-1)\sum_{l \not= i,N} d_{N, l} d_{i, l} \bigg]    =0,\ \ \ \ 1\leq i\not= j \leq N-1.
\end{equation}
If there is an $i$ for which the expression in the brackets in \eqref{DetEq} is nonzero, then $d_{N,j}=d_{i,j}$ for all $j\not= i,N$, which implies that ${\bf x}_{N} \sim {\bf x}_{i}$, which contradicts our assumption in the beginning of the proof. Therefore,
\begin{equation}\label{ScalarEq}
2N(2N-d_{N,i})-(N-1)\sum_{l \not= i,N} d_{N,l} d_{i,l}=0, \ \ \ i=1,\dots,N-1.
\end{equation}
Adding \eqref{ScalarEq} for $i=1,\dots,N-1$ we get using \eqref{sum_2N}
\begin{eqnarray}
0
&=&
2N\sum_{i=1}^{N-1} (2N-d_{N,i}) -(N-1)\sum_{i=1}^{N-1}\sum_{l \not= i,N} d_{N,l} d_{i,l} \nonumber \\
&=& (2N)^2 (N-1)-(2N)^2-(N-1)\sum_{l=1}^{N-1} d_{N,l} \sum_{i \not= l,N} d_{i,l} \nonumber \\
&=&
(2N)^2 (N-1)-(2N)^2-(N-1)\sum_{l=1}^{N-1} d_{N,l} (2N-d_{N,l} ) \nonumber \\
&=&
-(2N)^2+(N-1)\sum_{l=1}^{N-1} d_{N,l}^2.
\end{eqnarray}
In view of \eqref{sum_2N}, we find that equality holds in the Arithmetic-Quadratic Mean Inequality
\[
\bigg( \sum_{l=1}^{N-1} d_{N,l} \bigg)^2=(N-1)\sum_{l=1}^{N-1} d_{N,l}^2,
\]
which is possible only when $d_{N,1}=\dots=d_{N,N-1}$, which implies that (b) holds. This proves the Characterization Theorem.
\end{proof}

Next we continue with the proof of the Monotonicity Theorem, which implies that the only degenerate log-optimal (and $s$-energy optimal) configurations may be regular simplexes embedded in a sphere of higher dimension.
\begin{proof}[Proof of Theorem \ref{monotonicity}]
If $N\leq d+1$ the only optimal configuration is the regular $(N-1)$-simplex. This could be easily seen from (2.2) and the Geometric-Arithmetic Mean.
Indeed, for stationary configurations we have
\begin{equation}\label{mon_ineq1}
P(\omega_N)^2 = \prod_{i=1}^N \prod_{j\not= i}d_{i,j}\leq \prod_{i=1}^N \bigg( \Big(\sum_{j\not= i} d_{i,j}\Big)/(N-1)\bigg)^{N-1}=\bigg( \frac{2N}{N-1} \bigg)^{N(N-1)},
\end{equation}
and the upper bound is attained only if all the $d_{i,j}$'s are equal. But $N$ points lie in an $N-1$ dimensional hyperplane, which also must contain the origin (since it is a center of mass for stationary configurations), thus the optimal configuration lies in a $(N-1)$-dimensional subspace where the only $N$-point configuration with all mutual distances equal is the regular simplex. Since $N-1\leq d$, we can "fit" it in ${\Sp}^{d-1}$. This proves that $\mathcal{P}(N,d)=\mathcal{P}(N,N-1)$ for all $d\geq N$.

Now let $d<N$. It is clear that $\mathcal{P}(N,d-1)\leq \mathcal{P}(N,d)$ (the maximum over a larger set is larger). Then all we have to show is that a log-optimal configuration is non-degenerate. Indeed, if $\mathcal{P}(N,d)=\mathcal{P}(N,d-1)$ for some $d<N$, then there is an optimal configuration in ${\Sp}^{d-2}$ that is also an optimal configuration in ${\Sp}^{d-1}$, and thus is degenerate.

Suppose that $\omega_N= \{ {\bf x}_{1}, {\bf x}_{2} , \dots, {\bf x}_{N} \}$ is a log-optimal configuration in ${\Sp}^{d-1}$, which is degenerate. Then the hyperplane $G_{\omega_N}$ spanned by $\omega_N$ is of dimension $<d$. Because the center of mass ${\bf 0}$ is contained in $G_{\omega_N}$, we may assume that $\omega_N \subset \{ x_d =0 \}$. But ${\Sp}^{d-2}$ can no longer support the regular $(N-1)$-simplex (recall that $d<N$), so there is a pair of adjacent edges with unequal length, say $d_{1,3}<d_{2,3}$. Without loss of generality we can assume that

\[
{\bf x}_{1}=(r,\sqrt{1-r^2},0,\dots,0), \ \ \ {\bf x}_{2}= (r,-\sqrt{1-r^2},0,\dots,0).
\]
Consider the configuration $\omega_N' = \{ {\bf x}_{1 }', {\bf x}_{2}' , {\bf x}_{3 }, \dots, {\bf x}_{N} \}$, where
\[
{\bf x}_{1}'=(r,0,\dots,\sqrt{1-r^2}),\ \ \  {\bf x}_{2}'= (r,0,\dots,-\sqrt{1-r^2}).
\]
If ${\bf x}_{j}=(c_1,c_2,\dots,c_{d-1},0)$ is any point in $\omega_N$ (with $j\geq 3$), we have that

\begin{eqnarray} \label{mon_ineq2}
|{\bf x}_{j}-{\bf x}_{1}|^2 |{\bf x}_{j}-{\bf x}_{2}|^2 &=& (2-2c_1 r)^2-4c_2^2(1-r^2)\\
&\leq& (2-2c_1 r)^2=
|{\bf x}_{j}-{\bf x}_{1}'|^2 |{\bf x}_{j}-{\bf x}_{2}'|^2,\nonumber
\end{eqnarray}
with equality only if $c_2 =0$, which implies $d_{j,1}=d_{j,2}$. But for $j=3$ this is impossible and strict inequality holds in \eqref{mon_ineq2}. Since $ |{\bf x}_{1}-{\bf x}_{2}|=|{\bf x}_{1}'-{\bf x}_{2}'|$, we get $P(\omega_N) < P(\omega_N')$, which is a contradiction. This proves the theorem.
\end{proof}

\begin{remark}
The same argument can be applied to the Riesz $s$-energy case, namely \eqref{mon_ineq1} can be modified and equality will still hold for all distances equal, as well as for \eqref{mon_ineq2} we use the fact that the function
\[f(t)=(b-t)^{-s}+(b+t)^{-s}, \ b>t\geq 0,\]
achieves minimum when $t=0$ (maximum when $s<0$). Therefore, the conclusion of Theorem \ref{monotonicity} is true for  $\mathcal{E}_s (N,d)$ and the $s$-energy optimal points (see \eqref{s-energy}).
\end{remark}

\section{Three auxiliary results}
\setcounter{equation}{0}

Our first auxiliary Lemma deals with the case when condition (b) of Theorem \ref{characterization} holds.

\begin{lemma} \label{square-based}
Suppose $N=d+2$ and $\omega_N$ is a stationary logarithmic configuration that has a vertex with all outgoing edges equal. Suppose further that the log-optimal configuration of $d+1$ points on $\Sp^{d-2}$ satisfies Conjecture \ref{conjecture}. Then $\omega_N$ is not log-optimal and moreover,
\[ P(\omega_N)< P(\omega_{\{[ d/2 ],[ (d+1)/2 ]\}}).\]
\end{lemma}

\begin{proof}
Let $\omega_N$ be an optimal configuration for which Theorem \ref{characterization}(b) holds. Without loss of generality we may assume that ${\bf x}_{N}=(0,\dots0,1)$ and $d_{N,1}=d_{N,2}=\dots=d_{N,N-1}=2N/(N-1)$. Since ${\bf 0}$ is the center of mass of $\omega_N$, we have that
$\Omega_{N-1} :=\{ {\bf x}_{1},\dots,{\bf x}_{{N-1}} \} \subset \{ x_d = -1/(N-1) \}$.
Let $\tau_N := \{ {\bf y}_{1},\dots,{\bf y}_{{N-1}}, {\bf x}_{N} \}$ be an arbitrary configuration on
${\Sp}^{d-1}$ with $T_{N-1}:=\{ {\bf y}_{1},\dots,{\bf y}_{{N-1}} \}  \subset \{ x_d = -1/(N-1) \}$.
Then $P(\tau_N ) \leq P(\omega_N )$, and hence $P(T_{N-1} )\leq P(\Omega_{N-1})$. Thus, $\Omega_{N-1}$ is an optimal configuration in ${\Sp}^{d-1} \cap \{ x_d = -1/(N-1) \}$, which is a sphere in ${\R}^{d-1}$ of radius $r$, where $r^2=N(N-2)/(N-1)^2$.  Therefore,
\begin{equation}
P(\omega_N)
= \bigg( \frac{2N}{N-1} \bigg)^{N-1}  P(\Omega_{N-1} ) = 2^{N-1} \bigg( \frac{N}{N-1} \bigg)^{N-1} r^{(N-1)(N-2)}  \mathcal{P} (d+1,d-1) )
\end{equation}
We will compare $P(\omega_N)$ and $P(\omega_{\{[ d/2 ],[ (d+1)/2 ]\}})$.

Let $d=2k$ (the case $d=2k+1$ is being similar). By the assumption of the lemma the configuration $\Omega_{N-1}$ consists of two orthogonal regular $[(d-1)/2]$- and $[d/2]$-simplexes and hence formula \eqref{conj1} holds.

\begin{eqnarray}\label{squarebaseconf}
P(\omega_N) &= &2^{(k+1)(2k+1)}
\bigg( \frac{2k+2}{2k+1} \bigg)^{2k+1}
\bigg( \frac{2k(2k+2)}{(2k+1)^2} \bigg)^{2k(2k+1)/2}\\
& &\ \ \ \ \ \ \ \ \ \ \ \ \ \ \ \ \ \ \ \ \ \ \ \ \ \ \ \ \times \,\bigg(\frac{k}{k-1}\bigg)^{k(k-1)/2}
\bigg(\frac{k+1}{k}\bigg)^{k(k+1)/2}\nonumber
\end{eqnarray}

The quantity $P(\omega_{\{[ d/2 ],[ (d+1)/2 ]\}})$ is the right-hand side of \eqref{conj1} for $d=2k$, which simplifies to
\begin{equation}\label{bipyramidconf}
P(\omega_{\{[ d/2 ],[ (d+1)/2 ]\}}) = 2^{(k+1)(2k+1)}
\bigg(\frac{k+1}{k}\bigg)^{k(k+1)}
\end{equation}
We claim that $P(\omega_N)<P(\omega_{\{[ d/2 ],[ (d+1)/2 ]\}})$. Comparing \eqref{squarebaseconf} and \eqref{bipyramidconf} we have to verify the inequality
\begin{equation}\label{conditionBineq}
\bigg( \frac{2k}{2k+1} \bigg)^{(2k)(2k+1)/2}
\bigg( \frac{2k+2}{2k+1} \bigg)^{(2k+1)(2k+2)/2}<
\bigg(\frac{k-1}{k}\bigg)^{(k-1)k/2}
\bigg(\frac{k+1}{k}\bigg)^{k(k+1)/2},
\end{equation}
for all $k\geq 2$. Let
\begin{equation} \label{F}
F(x):= x(x+1)[\ln x -\ln (x+1)]
\end{equation} and $G(x):= F(x)-F(x+1)$. Then \eqref{conditionBineq} will hold if and only if $G(2k)<G(k-1)$. We differentiate $F(x)$ to find
\begin{eqnarray}
F'(x) &=& (2x+1) [\ln x -\ln(x+1)] +1 \label{firstderivative}\\
F'' (x) &=& 2[ \ln x -\ln(x+1) ]+ \frac{1}{x} +\frac{1}{x+1} \label{secondderivative} \\
F''' (x) &=& \frac{2}{x(x+1)} - \frac{1}{x^2} -\frac{1}{(x+1)^2}= -\bigg( \frac{1}{x} -\frac{1}{x+1}\bigg)^2 < 0. \label{thirdderivative}
\end{eqnarray}
From \eqref{thirdderivative} we get that $F'' (x)$ is strictly decreasing on $[1,\infty)$. Since $\lim_{x\to \infty} F'' (x) =0$, we derive that $F'' (x)>0$, and thus $F' (x)$ is strictly increasing on $[1,\infty)$. Since $G' (x) = F' (x)-F'(x+1)$ we finally conclude that $G' (x)<0$, and therefore $G(x)$ is strictly decreasing. This verifies $G(2k)<G(k-1)$ and \eqref{conditionBineq} for $k\geq 2$.

The case $d=2k+1$ is similar and reduces to $G(2k+1)<G(k)$, which of course also holds.
\end{proof}

In Proposition \ref{equivalence} we showed that the mirror relation ${\bf x}_i \sim{\bf x}_j$ is an equivalence relation. Moreover, the classes of equivalence form regular simplexes. Hence, if a configuration satisfies condition (c) of Theorem \ref{characterization}, then a natural decomposition of the configuration in regular simplexes  (components) occurs. We next show that if the hyperplane spanned by the points in such a component contains the origin (see \eqref{hyperplane}), then a configuration may be optimal only if it is $\omega_{\{[ d/2 ],[ (d+1)/2 ]\}}$.

\begin{lemma} \label{masscenter}
Suppose $N=d+2$ and let $\omega_N$ be a log-optimal configuration that satisfies condition (c) of Theorem \eqref{characterization}. If ${\bf 0}\in G_{U}$, where $U$ is a regular simplex component of $\omega_N$, then $\omega_N=\omega_{\{[ d/2 ],[ (d+1)/2 ]\}}$.
\end{lemma}

\begin{proof}
Suppose that $U:=\{ {\bf x}_{1}, \dots , {\bf x}_{{k+1}} \}$ is such a component and let ${\bf u}_{j}:={\bf x}_{1}-{\bf x}_{j},\ j=2,\dots,k+1$. First, we note that the vertices of a regular simplex are geometrically independent, i.e. the $\{ {\bf u}_{j} \}$'s are linearly independent.
Let $ V := \{ {\bf x}_{{k+2}}, \dots, {\bf x}_{{d+2}} \}$ and let ${\bf p}:=({\bf x}_{1}+\cdots + {\bf x}_{{k+1}})/(k+1)$ and ${\bf q}:=({\bf x}_{{k+2}}+\cdots + {\bf x}_{{d+2}})/(d-k+1)$ be the centers of mass of $U$ and $V$ respectively. Since ${\bf x}_{1} \sim {\bf x}_{j}$ for any $2 \leq j \leq k+1$, by Remark 2.3 we have that the radius-vectors $\{ {\bf x}_{s} \}$ are orthogonal to ${\bf u}_{j}$ for all $2\leq j \leq k+1$, $k+2 \leq s \leq d+2$.

By the assumptions of the lemma ${\bf p}={\bf 0}$, so $G_U$ is a $k$-dimensional subspace and all ${\bf x}_{j}\in V$ belong to its orthogonal complement. Since $d_{i,j} =2$ whenever $i\leq k+1 <j$,
and $\omega_N$ is log-optimal, $V$  must be a regular simplex also (as a log-optimal sub-configuration itself). We calculate the product
\begin{eqnarray}
P(\omega_N)
&=&
2^{(k+1)(d-k+1)} \bigg( \frac{2(k+1)}{k}\bigg)^{(k+1)k/2}
\bigg( \frac{2(d-k+1)}{d-k}\bigg)^{(d-k+1)(d-k)/2} \nonumber \\
&=& 2^{(d+1)(d+2)/2} \bigg( \frac{k+1}{k}\bigg)^{(k+1)k/2}
\bigg( \frac{d-k+1}{d-k}\bigg)^{(d-k+1)(d-k)/2}
\end{eqnarray}
We have that $\log(P(\omega_N)^2)=-F(k)-F(d-k)+const$, where $F(x)$ is defined in \eqref{F}. Let $H(x)=-F(x)-F(d-x)$. Then $H'(x)=-F'(x)+F'(d-x)$. But we found earlier (see \eqref{firstderivative} and the discussion after) that $F'(x)$ is strictly increasing on $[1,\infty)$, therefore $H'(x)>0$ on $[1,d/2)$ and $H'(x)<0$ on $(d/2,d-1]$ (clearly $H'(d/2)=0$). Thus, the maximum of $P(\omega_N )$ is achieved when $ k= [ d/2 ]$, which is the configuration $\omega_{\{[ d/2 ],[ (d+1)/2 ]\}}$.
\end{proof}

Our next Lemma shows that a stationary logarithmic configuration that is decomposed into three simplexes is {\bf not} log-optimal.

\begin{lemma} \label{three-simplex}
Suppose $N=d+2$ and $\Omega_N=U\cup V \cup W$ is a stationary logarithmic configuration that consists of three regular simplexes $U=\{ {\bf x}_1\dots, {\bf x}_k\}$, $V=\{ {\bf y}_1\dots, {\bf y}_l\}$, and $W=\{ {\bf z}_1\dots, {\bf z}_m\}$, $k+l+m=N$. Then $\Omega_N=\Omega_{\{k,l,m\}}$ is not log-optimal.
\end{lemma}

\begin{proof}
Suppose $\Omega_N$ is log-optimal, and let ${\bf u}$, ${\bf v}$, and ${\bf w}$ be centers of mass of the three simplexes $U$, $V$, and $W$ respectively. By Lemma \ref{masscenter} we have that ${\bf u}$, ${\bf v}$, and ${\bf w}$ are all non-zero vectors. We can utilize the representations
\begin{equation}\label{massEq}{\bf u}=\frac{1}{k}\sum_{i=1}^k {\bf x}_i,\quad {\bf v}=\frac{1}{l}\sum_{i=1}^l {\bf y}_i,\quad {\bf w}=\frac{1}{m}\sum_{i=1}^m {\bf z}_i,
\end{equation}
and the mirror relations to derive that
\[{\bf u}, {\bf v}, {\bf w} \in {\rm span}\left( \{ {\bf x}_1-{\bf x}_j\}_{j=2}^k,\{ {\bf y}_1-{\bf y}_j\}_{j=2}^l,\{ {\bf z}_1-{\bf z}_j\}_{j=2}^m \right)^\perp.\]
Since the vectors in the span are linearly independent, the span will be $k+l+m-3=d-1$ dimensional. Hence, its orthogonal complement is one dimensional and thus  ${\bf u}, {\bf v}, {\bf w}$ must be collinear. Denote by $u,v,w$ their coordinates w.r.t. a unit vector along the common line. The center of mass condition implies
\begin{equation}\label{uvw_masscenter} ku+lv+mw=0.
\end{equation}
There will be six distances, the edges of $U$, $V$, and $W$,
\[
|{\bf x}_i-{\bf x}_j|^2=\frac{2k(1-u^2)}{k-1},\ |{\bf y}_i-{\bf y}_j|^2=\frac{2l(1-v^2)}{l-1}, \ |{\bf z}_i-{\bf z}_j|^2=\frac{2m(1-w^2)}{m-1},
\]
and the distances between vertices in different simplexes
\[
|{\bf x}_i-{\bf y}_j|^2=2(1-uv),\ |{\bf x}_i-{\bf z}_j|^2=2(1-uw), \ |{\bf y}_i-{\bf z}_j|^2=2(1-vw).
\]
Indeed, let us illustrate how to derive one of these distance formulas, say
\[|{\bf x}_i-{\bf y}_j|^2=2(1-{\bf x}_i\cdot{\bf y}_j)=2(1-{\bf u}\cdot{\bf y}_j)=2(1-{\bf u}\cdot {\bf v})=2(1-uv),\]
where we used $({\bf x}_j-{\bf u})\cdot{\bf y}_j=0$ and ${\bf u}\cdot ({\bf y}_j-{\bf v})=0$, which easily follow from ${\bf x}_i\sim {\bf x}_j$, ${\bf y}_i\sim {\bf x}_j$, and \eqref{massEq}.

From ${\rm rank}(M)=2$ and \eqref{CoeffEq} we have (when $i=1$) that
\[ \frac{d_{N,2}}{d_{1,2}}=(N-1)-\sum_{j\not=1,N} \frac{d_{N,j}}{d_{1,j}}
= d_{N, 1}\bigg( \frac{N-1}{2}-\sum_{j\not= 1} \frac{1}{d_{i,j}}\bigg) .\]
Simplification (recall that $d_{N,2}=d_{N,1}$ as ${\bf x}_1\sim {\bf x}_2$), and similar considerations for the remaining indexes $i$ in \eqref{CoeffEq} yield the following equations
\begin{eqnarray*}
\frac{k-1}{2(1-u^2)}+\frac{l}{2(1-uv)}+\frac{m}{2(1-uw)}&=&\frac{N-1}{2},\\
\frac{k}{2(1-uv)}+\frac{l-1}{2(1-v^2)}+\frac{m}{2(1-vw)}&=&\frac{N-1}{2},\\
\frac{k}{2(1-uw)}+\frac{l}{2(1-vw)}+\frac{m-1}{2(1-w^2)}&=&\frac{N-1}{2}.
\end{eqnarray*}
Algebraic manipulations give the system (we use that $u,v,w \not=0$)
\begin{eqnarray}
\frac{(k-1)u}{2(1-u^2)}+\frac{lv}{2(1-uv)}+\frac{mw}{2(1-uw)}&=&0,\label{Eq_u}\\
\frac{ku}{2(1-uv)}+\frac{(l-1)v}{2(1-v^2)}+\frac{mw}{2(1-vw)}&=&0,\label{Eq_v}\\
\frac{ku}{2(1-uw)}+\frac{lv}{2(1-vw)}+\frac{(m-1)w}{2(1-w^2)}&=&0\label{Eq_w}.
\end{eqnarray}
Substituting $mw=-ku-lv$ from \eqref{uvw_masscenter} into \eqref{Eq_u} and \eqref{Eq_v} we obtain (after dividing by $u$ and $v$ respectively)
\begin{eqnarray*}
(ku^2-(k-1)uw-1)(1-uv)+lv(v-w)(1-u^2)&=&0 \\
(lv^2-(l-1)vw-1)(1-uv)+ku(u-w)(1-v^2)&=&0,
\end{eqnarray*}
which after subtraction reduces to
\[w(u-v)[(k+l+m-1)uv+1]=0.\]

Suppose that all $u,v,w$ are distinct. Then $uv=-1/(k+l+m-1)$. By symmetry we obtain $uw=-1/(k+l+m-1)$, thus deriving $u(v-w)=0$, which is absurd as $u,v,w\not=0$ and $u,v,w$ are distinct.

Therefore, without loss of generality we may assume that $u=v$. Since $w=-(k+l)u/m$, we have $u\not= w$, which as in the derivation above implies
\[ uw=-\frac{1}{k+l+m-1}.\]
Together with \eqref{sum_2N} this yields that
\[|{\bf z}_i-{\bf z}_j|^2=|{\bf z}_i-{\bf x}_j|^2=|{\bf z}_i-{\bf y}_j|^2=\frac{2N}{N-1}.\]
From $u=-mw/(k+l)$ we get $u^2=-muw/(k+l)=m/[(k+l)(k+l+m-1)]$, which gives
\[ 1-u^2 = \frac{(k+l-1)(k+l+m)}{(k+l)(k+l+m-1)}. \]

We now are ready to compute
\[
P(\Omega_{\{k,l,m\}})=\frac{2^{\frac{N(N-1)}{2}}\left( \frac{N}{N-1}\right)^{\frac{N(N-1)}{2}}\left( \frac{k}{k-1}\right)^{\frac{k(k-1)}{2}}\left( \frac{l}{l-1}\right)^{\frac{l(l-1)}{2}}}{\left( \frac{k+l}{k+l-1}\right)^{\frac{(k+l)(k+l-1)}{2}}}.\]

We shall compare this product with the product of a two-simplex $\omega_{\{k,l+m\}}$, where the origin is the center of mass for both simplexes. We have that

\[
P(\omega_{\{k,l+m\}})=2^{\frac{N(N-1)}{2}}\left( \frac{k}{k-1}\right)^{\frac{k(k-1)}{2}}\left( \frac{l+m}{l+m-1}\right)^{\frac{(l+m)(l+m-1)}{2}}.\]

The inequality $P(\Omega_{\{k,l,m\}})<P(\omega_{\{k,l+m\}})$ is equivalent to
\begin{equation} \label{prod_ineq}
\left( \frac{l}{l-1}\right)^{\frac{l(l-1)}{2}}\left( \frac{N}{N-1}\right)^{\frac{N(N-1)}{2}}<\left( \frac{k+l}{k+l-1}\right)^{\frac{(k+l)(k+l-1)}{2}}\left( \frac{l+m}{l+m-1}\right)^{\frac{(l+m)(l+m-1)}{2}}.\end{equation}
Let
\[L(x):=\frac{x(x-1)}{2}[\ln x - \ln(x-1)].\]
Observe, that $L(x)=-F(x)/2$, where $F(x)$ is defined by \eqref{F}. From \eqref{secondderivative} and the discussion therein we have that $L''(x)<0$. But the inequality \eqref{prod_ineq} is equivalent to
\[ L(l)+L(k+l+m)<L(l+m)+L(k+l) ,\]
which can be easily seen from the concavity property of $L(x)$. Indeed, the chord connecting $(l,L(l))$ and $(k+l+m,L(k+l+m))$ lies below the graph (a drawing may be beneficial). If we denote the intersections of the chord with the vertical lines $x=l+m$ and $x=k+l$ as $(l+m,M_1)$ and $(k+l,M_2)$, then
\[ L(l)+L(k+l+m)=M_1+M_2<L(l+m)+L(k+l).\]
This completes the proof of Lemma \ref{three-simplex}.
\end{proof}

\section{Two new log-optimal configurations}
\setcounter{equation}{0}

We now proceed with the proofs of the log-optimality of the two configurations in Theorem \ref{log-optimal}.

\begin{proof}[Proof of Theorem \ref{log-optimal}] i) In this case all conditions (a), (b) and (c) of Theorem \ref{characterization} are possible. The degenerate configuration with minimal logarithmic energy of dimension two is the regular hexagon and  in dimension three is the octahedron (see \cite{KY}). Theorem \ref{monotonicity} implies that none of these are log-optimal configurations on $\Sp^3$.

When (b) holds, without loss of generality we may assume that the North Pole $(0,0,0,1)$ is equidistant to the other five points which are contained in the hyperplane $x_4=-1/5$. The configuration of this type that minimizes energy has two diametrically opposite points $(0,0, \pm \sqrt{24}/5,-1/5)$ and an equilateral triangle orthogonal to that diameter $\{ (\sqrt{24}/5\cos\frac{2k\pi}{3},\sqrt{24}/5\sin\frac{2k\pi}{3},0,-1/5)\}$ (see Corollary \ref{bipyramidOpt}), i.e.
\[\omega_6 = \{(0,0,0,1),(0,0,\pm \sqrt{24}/5,-1/5),(\sqrt{24}/5\cos\frac{2k\pi}{3},\sqrt{24}/5\sin\frac{2k\pi}{3},0,-1/5)\}\]
That this is not log-optimal follows from Lemma \ref{square-based}.

The situation when (c) holds is richer. The various equivalence classes under the mirror relation give rise to the following configurations:

\begin{itemize}
\item[A.] Two orthogonal simplexes, a diameter and regular tetrahedron, with $2$ and $4$ points respectively;
\[ \omega_{\{2,4\}}=\left\{(0,0,0,\pm 1)\right\}\cup\left \{(1,0,0,0),(-\frac{1}{3},\frac{2\sqrt{2}}{3}\cos\frac{2k\pi}{3},\frac{2\sqrt{2}}{3}\sin\frac{2k\pi}{3},0)\right\}_{k=0}^2 \]
Lemma \ref{masscenter} implies that this is not log-optimal.
\item[B.] Three orthogonal simplexes
\[\{(u,\pm \sqrt{1-u^2},0,0)\}\cup\{(v,0,\pm \sqrt{1-v^2},0)\}\cup\{(w,0,0,\pm \sqrt{1-w^2})\},\ \ u+v+w=0.\]
This is not log-optimal because of Lemma \ref{three-simplex}.
\item[C.]  Two orthogonal simplexes with $3$ points each (equilateral triangles);
\[\omega_{\{3,3\}}=\left\{(\cos\frac{2k\pi}{3},\sin\frac{2k\pi}{3},0,0)\right\}_{k=0}^2 \cup \left\{(0,0,\cos\frac{2k\pi}{3},\sin\frac{2k\pi}{3})\right\}_{k=0}^2.\]
This is the log-optimal configuration of six points on $\Sp^3$.
\end{itemize}

ii) The degenerate configuration with minimal logarithmic energy of dimension two is the regular heptagon and dimension  three is not known, conjectured by Rakhmanov to be two diametrically opposite points and a regular pentagon on the equatorial circle. Theorem \ref{monotonicity} implies that these are not log-optimal configurations.

When (b) holds, without loss of generality we may assume that the north pole $(0,0,0,0,1)$ is equidistant to the other six points which are contained in the hyperplane $x_5=-1/6$. The configuration of this type that minimizes energy will have to have two orthogonal equilateral triangles \[T_1=\{(\sqrt{35}/6\cos\frac{2k\pi}{3},\sqrt{35}/6\sin\frac{2k\pi}{3},0,0,-1/6)\}_{k=0}^2\] and \[T_2=\{(0,0,\sqrt{35}/6\cos\frac{2k\pi}{3},\sqrt{35}/6\sin\frac{2k\pi}{3},-1/6)\}_{k=0}^2\] on the hyperplane $x_5=-1/6$. That this is not log-optimal follows from Lemma \ref{square-based} and part i).

The situation when (c) holds is similar to the previous case. The various equivalence classes under the mirror relation give rise to the following configurations:

\begin{itemize}
\item[A.] Two orthogonal simplexes, $\omega_{\{2,5\}}$, a diameter and regular $5$-point simplex on $\Sp^3$.
Lemma \ref{masscenter} implies that this is not log-optimal.
\item[B.] Three orthogonal simplexes $\omega_{\{2,2,3\}}$.
This is not log-optimal because of Lemma \ref{three-simplex}.

\item[C.]  Two orthogonal simplexes $\omega_{\{3,4\}}$, with $3$ points (equilateral triangle) and $4$ points (regular tetrahedron), respectively.
This is the log-optimal configurations of seven points on $\Sp^4$.
\end{itemize}
\end{proof}

{\bf Acknowledgement}: I would like to thank the honoree for his great influence and continued support throughout the span of my career.

\bibliographystyle{amsalpha}

\end{document}